\documentclass[10pt]{article}
\usepackage{amsmath,amsthm,amsfonts,amssymb,latexsym,graphicx}
\usepackage[pdfpagemode=UseNone,pdfstartview=FitH]{hyperref}

\makeatletter

\newif\iftwodates
\twodatesfalse

\renewcommand\maketitle{\begin{titlepage}%
  \pagenumbering{Alph}%
  \let\footnotesize\small
  \let\footnoterule\relax
  \let\footnote\thanks
  \null\vfil
  \vskip 30\p@
  \begin{center}%
    {\LARGE \bf \@title \par}%
    \vskip 3em%
    {\large
     \lineskip .75em%
     \begin{tabular}[t]{c}%
       \@author
     \end{tabular}\par}%
     \vskip 1.5em%
  \end{center}\par
  \vfill
  \begin{center}
    \raisebox{1.5cm}{\includegraphics[width=0.58\textwidth]%
      {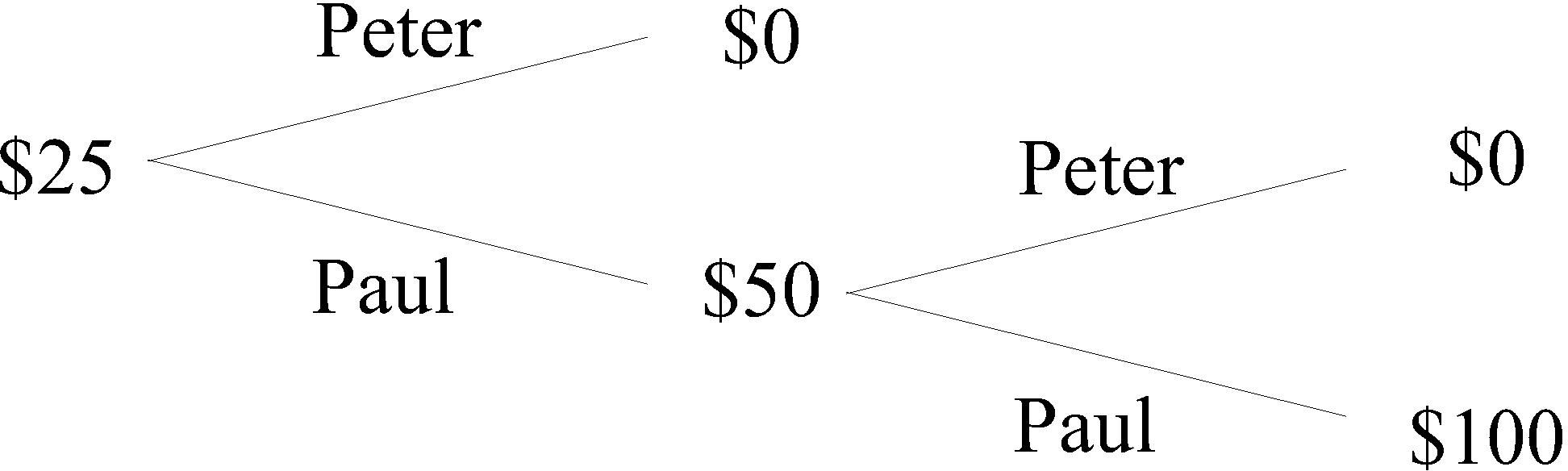}}%
    \hskip 3em%
    \includegraphics[width=0.29\textwidth]%
      {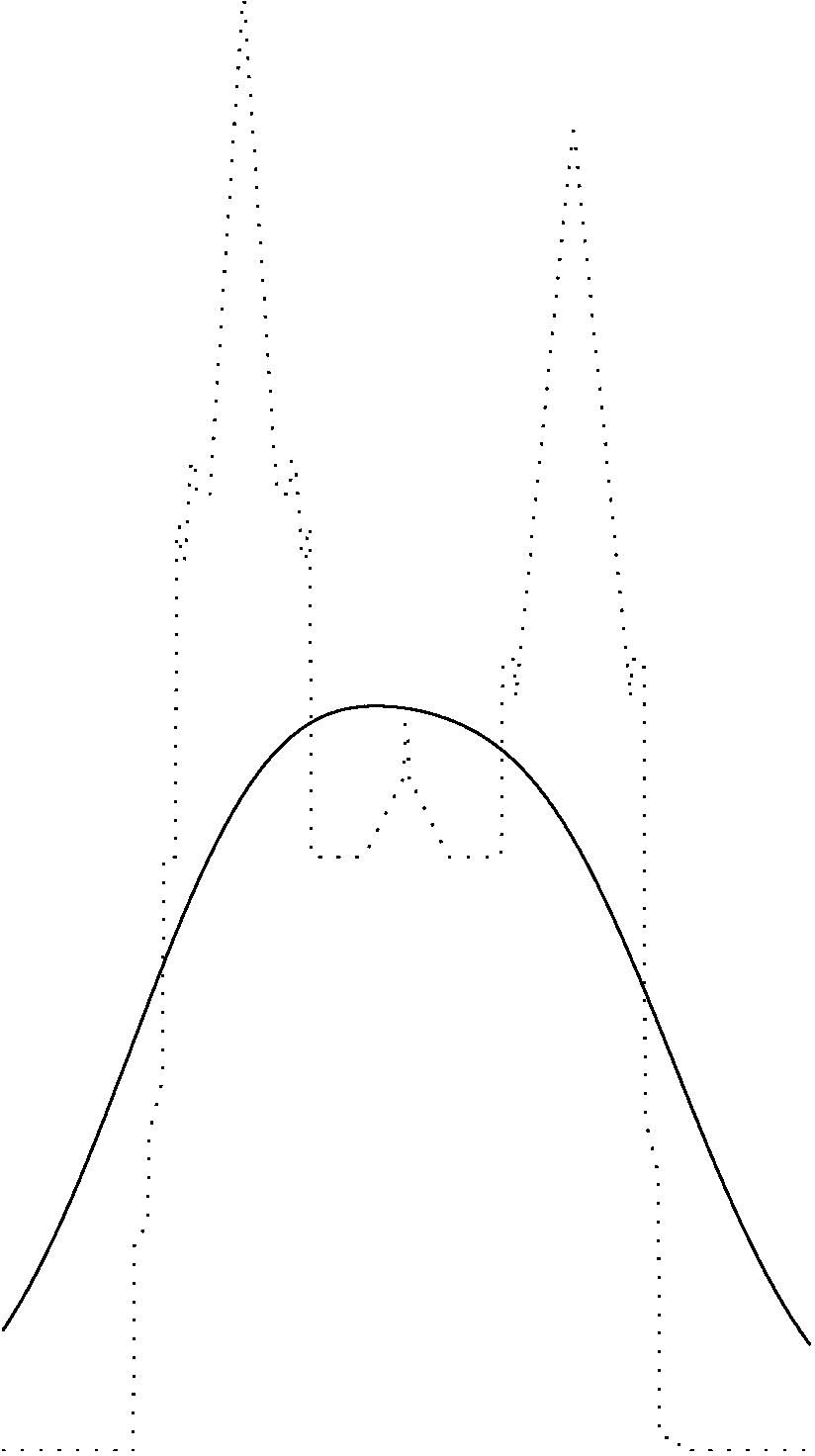}%
  \end{center}
  \@thanks
  \vfill
  \begin{center}
    {\large \bf The Game-Theoretic Probability and Finance Project}
  \end{center}
  \begin{center}
    {\large Working Paper \#\No}
  \end{center}
  \begin{center}
    {\iftwodates\large First posted \firstposted.
    Last revised \@date.\else\large\@date\fi}
  \end{center}
  \begin{center}
    Project web site:\\
    http://www.probabilityandfinance.com
  \end{center}
  \end{titlepage}%
  \setcounter{footnote}{0}%
  \global\let\thanks\relax
  \global\let\maketitle\relax
  \global\let\@thanks\@empty
  \global\let\@author\@empty
  \global\let\@date\@empty
  \global\let\@title\@empty
  \global\let\title\relax
  \global\let\author\relax
  \global\let\date\relax
  \global\let\and\relax
}

\renewenvironment{abstract}{%
  \titlepage\pagenumbering{roman}%
  \null\vfil
  \@beginparpenalty\@lowpenalty
  \begin{center}%
    \Large \bfseries \abstractname
    \@endparpenalty\@M
  \end{center}}%
  {\par\vfill\tableofcontents\thispagestyle{empty}\endtitlepage
  \pagenumbering{arabic}}

\renewenvironment{thebibliography}[1]
  {\section*{\refname}%
  \addcontentsline{toc}{section}{\refname}%
  \@mkboth{\MakeUppercase\refname}{\MakeUppercase\refname}%
  \list{\@biblabel{\@arabic\c@enumiv}}%
    {\settowidth\labelwidth{\@biblabel{#1}}%
    \leftmargin\labelwidth
    \advance\leftmargin\labelsep
    \@openbib@code
    \usecounter{enumiv}%
    \let\p@enumiv\@empty
    \renewcommand\theenumiv{\@arabic\c@enumiv}}%
    \sloppy
    \clubpenalty4000
    \@clubpenalty \clubpenalty
    \widowpenalty4000%
    \sfcode`\.\@m}
    {\def\@noitemerr
    {\@latex@warning{Empty `thebibliography' environment}}%
  \endlist}

\makeatother

\allowdisplaybreaks[4]

\newcommand{\GGG}{\mathcal{G}}
\newcommand{\HHH}{\mathcal{H}}

\newcommand{\MMM}{\mathcal{M}}
\newcommand{\SSS}{\mathcal{S}}

\newcommand{\bbbr}{\mathbb{R}}
\newcommand{\bbbp}{\mathbb{P}}

\newcommand{\LP}{\mathop{\underline{\bbbp}}\nolimits}

\newtheorem{proposition}{Proposition}

\theoremstyle{definition}

\newlength{\IndentI}
\newlength{\IndentII}
\newlength{\IndentIII}
\newlength{\WidthI}
\newlength{\WidthII}
\newlength{\WidthIII}
\title{Game-Theoretic Capital Asset Pricing in Continuous Time}%
\author{Vladimir Vovk and Glenn Shafer}

\newcommand{\No}{2}
\twodatestrue
\newcommand{\firstposted}{December 30, 2001}

\begin{document}
\maketitle
\begin{abstract}
  We derive formulas for the performance of capital assets
  in continuous time from an efficient market hypothesis,
  with no stochastic assumptions
  and no assumptions about the beliefs or preferences of investors.
  Our efficient market hypothesis says that
  a speculator with limited means
  cannot beat a particular index by a substantial factor.
  Our results include a formula that resembles the classical CAPM formula
  for the expected simple return of a security or portfolio.

  This version of the article was essentially written in December 2001
  but remains a working paper.
\end{abstract}

In this article, we use an efficient market hypothesis
to derive formulas for the performance of capital assets in continuous time.
Our efficient market hypothesis says that a speculator with limited means
cannot beat a particular market index $m$ by a substantial factor.
From this hypothesis, we derive two results
concerning the average returns of a security or portfolio $s$:
\begin{itemize}
\item
   The average simple return for $s$ will fall short
   of the average simple return for $m$
   by an amount equal to the difference
   between the variance of $m$'s simple returns
   and the covariance of $m$'s and $s$'s simple returns.
   (See the formula in Proposition~\ref{prop:1}.)
   This agrees with the classical capital asset pricing model,
   except that the classical model relates
   theoretical centered moments of a probability distribution,
   whereas our result is about empirical uncentered moments.
\item
   The average logarithmic return for $s$
   will fall short of the average logarithmic return for $m$
   by an amount equal to one-half the variance
   of the difference between the simple returns.
   (See the formula in Proposition~\ref{prop:2}.)
\end{itemize}
These results follow from the efficient market hypothesis alone,
with no stochastic assumptions
and no assumptions about the beliefs or preferences of investors.

Our results are for continuous time and use nonstandard analysis.
They parallel the results for discrete time (with explicit error bounds)
in our article \cite{GTP1}.
That article gives more detail on motivation and possible applications.
For an alternative treatment of continuous time
(more traditional and not using nonstandard analysis),
see \cite{GTP44} and \cite{GTP45}.

\section{Definitions}

Our results depend on the explicit formulation of a game.  
Here we begin by describing the game informally.
Then, after introducing notation for the moments of 
$s$ and $m$'s returns, we state the protocol for the
game formally and explain how it can be used to formalize
the implications of our efficient market hypothesis.

\subsection{The Basic Capital Asset Pricing Game}

The capital asset pricing game has two principal players, 
Speculator and Market, who alternate play.
On each round,
Speculator decides how much of each security in the market to hold
(and possibly short),
and then Market determines Speculator's gain
by deciding how the prices of the securities change.
Allied with Market is a third player, Investor,
who also invests each day.
The game is a perfect-information game:
each player sees the others' moves.

We assume that there are $K+1$ securities in the market
and $N$ rounds (trading periods) in the game.
We number the securities from $0$ to $K$
and the rounds from $1$ to $N$, and we write $x_n^k$
for the simple return on security $k$ in round $n$.
For simplicity, we assume that $-1 < x_n^k < \infty$
for all $k$ and $n$;  a security price never becomes zero.
We write $x_n$ for the vector
$(x_n^0,\dots,x_n^K)$, which lies in $(-1,\infty)^{K+1}$.
Market determines the returns;
$x_n$ is his move on the $n$th round.

We assume that the first security, indexed by $0$, is our market index $m$.
If $m$ is a portfolio formed from the other securities,
then $x_n^0$ is an average of the $x_n^1,\dots,x_n^K$,
but we do not insist on this.

We write $\MMM_n$ for the capital at the end of round $n$
resulting from investing one monetary unit in $m$
at the beginning of the game:
\[
  \MMM_n
  :=
  \prod_{i=1}^n (1 + x_i^0).
\]
Thus $\MMM_N$ is the final capital resulting from this investment.

Investor begins with capital equal to one monetary unit
and is allowed to redistribute his current capital
across all $K+1$ securities on each round.
If we write $\GGG_n$ for his capital at the end of the $n$th round,
then
\[
  \GGG_n
  :=
  \prod_{i=1}^n
  \sum_{k=0}^K
  g_i^k (1+x_i^k),
\]
where $g_i^k$
is the fraction of his capital he holds in security $k$ during the $i$th round.
The $g_i^k$ must sum to~$1$ over~$k$,
but $g_i^k$ may be negative for a particular $k$
(in this case Investor is selling $k$ short).
Investor's final capital is $\GGG_N$.

We call the set of all possible sequences
$(g_1,x_1,\dots,g_N,x_N)$
the \emph{sample space} of the game,
and we designate it by $\Omega$:
\[
  \Omega
  :=
  \left(
    \bbbr^{K+1}\times(-1,\infty)^{K+1}
  \right)^N.
\]
We call any subset of $\Omega$ an \emph{event}.
Any statement about Investor's returns determines an event,
as does any comparison of Investor's and Market's returns.

Speculator also starts with one monetary unit
and is allowed to redistribute
his current capital across all $K+1$ securities on each round.
We write $\HHH_n$ for his capital at the end of the $n$th round:
\[
  \HHH_n
  :=
  \prod_{i=1}^n
  \sum_{k=0}^K
  h_i^k (1+x_i^k),
\]
where $h_i^k$ is the fraction of his capital he 
holds in security $k$ during the $i$th round.
The moves by Speculator are not recorded in the sample space;
they do not define events.

Our results use non-standard analysis.
We assume that the number $N$ of rounds of the game is 
infinitely large.  The game begins at time $0$.  Each
round takes an infinitesimal amount of time $dt$, and 
play ends at time $T$, which is an infinitely large 
positive real number:  $T=Ndt$.
A brief summary of nonstandard analysis,
sufficient for our purposes, is provided in
\cite{Shafer/Vovk:2001}, \S11.5;
further details can be found in \cite{Goldblatt:1998}.

\subsection{Notation for Moments}

Let us write $s_n$ for Investor's simple return on round $n$,
and $m_n$ for the simple return of
the market index $m$ on round $n$:
\[
       s_n
       :=
       \frac{\GGG_n-\GGG_{n-1}}{\GGG_{n-1}}
       =
       \sum_k g_n^k x_n^k,
\]
and 
\[
       m_n := x_n^0.
\]
For every play of the game
we define the following nonstandard numbers:
\[
  \mu_s
  =
  \frac1N
  \sum_{n=1}^N
  \frac{s_n}{dt}
  =
  \frac1T
  \sum_{n=1}^N
  s_n
\]
(this is the average rate of increase in Investor's capital),
\[
  \sigma_s^2
  =
  \frac1N
  \sum_{n=1}^N
  \frac{s_n^2}{dt}
  =
  \frac1T
  \sum_{n=1}^N
  s_n^2
\]
($\sigma_s$ is the empirical volatility of Investor's capital),
\[
  \mu_m
  =
  \frac1N
  \sum_{n=1}^N
  \frac{m_n}{dt}
  =
  \frac1T
  \sum_{n=1}^N
  m_n,
\quad
  \sigma_m^2
  =
  \frac1N
  \sum_{n=1}^N
  \frac{m_n^2}{dt}
  =
  \frac1T
  \sum_{n=1}^N
  m_n^2
\]
(analogous quantities for the index).
We also set
\begin{align*}
  \sigma_{sm}
  &=
  \frac1N
  \sum_{n=1}^N
  \frac{s_n m_n}{dt}
  =
  \frac1T
  \sum_{n=1}^N
  s_n m_n, \\
  \sigma_{s-m}^2
  &=
  \frac1N
  \sum_{n=1}^N
  \frac{(s_n-m_n)^2}{dt}
  =
  \frac1T
  \sum_{n=1}^N
  (s_n-m_n)^2
\end{align*}
and
\begin{align*}
  \lambda_s
  =
  \frac1N
  \sum_{n=1}^N
  \frac{\ln(1+s_n)}{dt}
  =
  \frac1T
  \sum_{n=1}^N
  \ln(1+s_n), \\
  \lambda_m
  =
  \frac1N
  \sum_{n=1}^N
  \frac{\ln(1+m_n)}{dt}
  =
  \frac1T
  \sum_{n=1}^N
  \ln(1+m_n)
\end{align*}
(the last two are the average logarithmic rates of growth).
Notice that $\exp(\lambda_s T)$ is the total relative increase $\GGG_N$
in Investor's capital
and $\exp(\lambda_m T)$ is the total relative increase $\MMM_N$
in the value of the index;
therefore,
$\lambda_s$ and $\lambda_m$ are direct measures
of Investor's and the index's performance.

As a simple example,
consider the case where Investor
always holds one share of a security whose price $S_t$ is generated by Market
using the stochastic differential equation
\[
  \frac{dS_t}{S_t}
  =
  \mu dt + \sigma dW_t,
\]
where $W_t$ is a Brownian motion.
In this case, $\mu_s$ will be infinitely close to $\mu$
almost surely,
and $\sigma_s$ will be infinitely close to $\sigma$
almost surely.

\subsection{The Protocol}

We can now state precisely the protocol for the game involving
Investor, Market, and Speculator:

\bigskip

\noindent
\textsc{Basic Capital Asset Pricing Protocol (Basic CAP Protocol)}

\noindent
\textbf{Players:} Investor, Market, Speculator

\noindent
\textbf{Parameters:}

\parshape=2
\IndentI  \WidthI
\IndentI  \WidthI
\noindent
Natural number $K$ (number of non-index securities in the market)\\
Infinite natural number $N$ (number of rounds or trading periods)

\noindent
\textbf{Protocol:}

\parshape=10
\IndentI  \WidthI
\IndentI  \WidthI
\IndentI  \WidthI
\IndentI  \WidthI
\IndentII \WidthII
\IndentII \WidthII
\IndentII \WidthII
\IndentII \WidthII
\IndentII \WidthII
\IndentII \WidthII
\noindent
$\GGG_0 := 1$.\\
$\HHH_0:=1$.\\
$\MMM_0 := 1$.\\
FOR $n=1,2,\dots,N$:\\
  Investor selects $g_n \in \bbbr^{K+1}$
        such that $\sum_{k=0}^K g_n^k = 1$.\\
  Speculator selects $h_n \in \bbbr^{K+1}$
        such that $\sum_{k=0}^K h_n^k = 1$.\\
  Market selects $x_n \in (-1,\infty)^{K+1}$.\\
  $\GGG_n := \GGG_{n-1} \sum_{k=0}^K g_n^k (1 + x_n^k)$.\\
  $\HHH_n := \HHH_{n-1} \sum_{k=0}^K h_n^k (1 + x_n^k)$.\\
  $\MMM_n := \MMM_{n-1}(1 + x_n^0)$.

\noindent
\textbf{Restriction:}

\noindent
Market and Investor are required to make
$\sigma_s^2$ and $\sigma_m^2$ finite
and to make
$\max_n |s_n|$ and $\max_n |m_n|$ infinitesimal.

\bigskip

\noindent
The condition that $\max_n |s_n|$ and $\max_n |m_n|$ be infinitesimal
is a continuity condition,
but there is a slight complication arising from the fact that $T$ is infinite:
if $T$ is extremely large as compared with $1/dt$,
the largest of the huge number of increments might become nonnegligible.
This, however, would be an extreme situation:
e.g., for the usual diffusion processes the condition $\ln T \le (dt)^{-1/2}$
is more than sufficient for the largest increment to be negligible.

\subsection{Predictions from the EMH for $m$}

We now adopt an efficient market hypothesis:  Market will 
not allow Speculator to become very rich relative to the index $m$.
Intuitively, this sometimes 
implies that a certain event $A$ will happen.
To formalize this intuition, let us 
say that \emph{the efficient market hypothesis for} $m$
\emph{predicts} $A$ \emph{at level} $\alpha>0$
if Speculator has a strategy $\SSS$ in the basic CAP protocol that ensures
the following:
$\HHH_n \ge 0$ for $n=1,\dots,N$
and
either (1) $\HHH_N \ge \frac{1}{\alpha} \MMM_N$
or (2) $(g_1,x_1,\dots,g_N,x_N) \in A$.
(It will be convenient to say that such a strategy $\SSS$
\emph{witnesses} that the efficient market hypothesis predicts $A$
at level $\alpha$.)

For brevity, we will abbreviate ``efficient market hypothesis for $m$''
to ``EMH for $m$''.
As the reader may have noticed, ``EMH for $m$'' is not a mathematical concept for us;
we have not given it a precise definition.
We have, however, provided a precise definition
for the phrase ``the EMH for $m$ predicts $A$ at level $\alpha>0$''.

Our confidence that Speculator will not beat the market by $\frac{1}{\alpha}$
is greater for smaller $\alpha$.
So a prediction of $A$ at level $\alpha$ becomes more emphatic
as $\alpha$ decreases.
The most emphatic prediction arises in the limit, 
when the EMH for $m$ predicts $A$ at every level $\alpha>0$.
In this case, we say simply that \emph{the EMH for} $m$
\emph{predicts} $A$.

\section{Results}

Proofs of the propositions in this section
will be provided in \S\ref{sec:proofs}.

\subsection{Capital Asset Pricing Model}

\begin{proposition}\label{prop:1}
  For any $\epsilon>0$,
  the EMH for $m$ predicts
  \[
    \left|
      \mu_s - \mu_m + \sigma_m^2 - \sigma_{sm}
    \right|
    <
    \epsilon
    (1 + \sigma_{s-m}^2).
  \]
\end{proposition}
If it is known \emph{a priori} that
$\sigma_{s-m}^2 < C$ for some positive constant $C$,
then for every $\epsilon>0$,
the EMH for $m$ predicts that
\[
    \left|
  \mu_s - \mu_m + \sigma_m^2 - \sigma_{sm}
    \right|
  <
  \epsilon.
\]
This is analogous to the capital asset pricing model (CAPM)
of the established theory.  Remarkably, we get the result without
the strong assumptions of that theory.  We assume nothing about 
Investor's beliefs or preferences, and we do not assume that 
Market chooses his prices stochastically.

\subsection{Theoretical Performance Deficit}

The next proposition justifies calling $\sigma_{s-m}^2/2$
the ``theoretical performance deficit''.
\begin{proposition}\label{prop:2}
  For any $\epsilon>0$,
  the EMH for $m$ predicts that
  \[
    \left|
      \lambda_s
      -
      \lambda_m
      +
      \frac12 \sigma_{s-m}^2
    \right|
    <
    \epsilon (1 + \sigma_{s-m}^2).
  \]
\end{proposition}
Again if it is known \emph{a priori} that
$\sigma_{s-m}^2 < C$ for some positive constant $C$,
then for every $\epsilon>0$
the EMH for $m$ predicts that
\[
  \left|
    \lambda_s
    -
    \lambda_m
    +
    \frac12 \sigma_{s-m}^2
  \right|
  <
  \epsilon.
\]
This result suggests that an analysis of the variance of 
the vector of differences
$(s_1 - m_1,\dots,s_N - m_N)$ might give
insight into the performance of the portfolio $s$.

\section{Proofs}
\label{sec:proofs}

\begin{proof}[Proof of Proposition~\ref{prop:1}]
  When $x>-1$,
  we can expand $\ln(1+x)$ in a Taylor's series with remainder:
  \begin{equation}\label{eq:Taylor}
    \ln(1+x)
    =
    x - \frac12 x^2 + \frac13 x^3 \frac{1}{(1+\theta x)^3},
  \end{equation}
  where $\theta$, which depends on $x$, satisfies $0\le\theta\le1$.
  Since
  \[
    \gamma(x)
    \le
    \frac13 x^3 \frac{1}{(1+\theta x)^3}
    \le
    \Gamma(x),
  \]
  where the functions $\gamma$ and $\Gamma$ are defined by
  \[
    \gamma(x)
    :=
    \frac13
    \left(
      \frac{x}{1+x}
    \right)^3,
  \quad
    \Gamma(x)
    :=
    \frac13 x^3,
  \]
  we can see that~\eqref{eq:Taylor} implies
  \[
    \ln(1+x)
    \le
    x - \frac12 x^2
    +
    \Gamma(x)
  \]
  and
  \[
    \ln(1+x)
    \ge
    x - \frac12 x^2
    +
    \gamma(x).
  \]
  Notice that the functions $\Gamma$ and $\gamma$ are monotonically increasing.

  On a few occasions we will use the identity
  \[
    \frac{\sigma_{s-m}^2}{2}
    =
    \frac{\sigma_s^2}{2}
    +
    \frac{\sigma_m^2}{2}
    -
    \sigma_{sm}.
  \]

  Now we are ready to start proving the proposition.
  First we split it into two:
  we will prove separately that the EMH for $m$ predicts
  \begin{equation}\label{eq:CAPM1}
    \mu_s - \mu_m + \sigma_m^2 - \sigma_{sm}
    <
    \epsilon (1 + \sigma_{s-m}^2)
  \end{equation}
  and that the EMH for $m$ predicts
  \begin{equation}\label{eq:CAPM2}
    \mu_s - \mu_m + \sigma_m^2 - \sigma_{sm}
    >
    -\epsilon (1 + \sigma_{s-m}^2).
  \end{equation}
  Indeed, if Speculator has a strategy witnessing
  that the EMH for $m$ predicts~\eqref{eq:CAPM1} at level $\alpha/2$
  and a strategy witnessing
  that the EMH for $m$ predicts~\eqref{eq:CAPM2} at level $\alpha/2$,
  the combination of these strategies
  (i.e., splitting his money equally into two accounts
  and letting the first account be managed according to the first strategy
  and the second account be managed according to the second strategy)
  witnesses that the EMH for $m$ predicts
  the conjunction of~\eqref{eq:CAPM1} and~\eqref{eq:CAPM2} at level $\alpha$.
  (This argument is essentially an implicit application of the inequality
  \[
    \LP(A\cap B)
    \ge
    \LP(A)
    +
    \LP(B)
    -
    1
  \]
  from \cite{Shafer/Vovk:2001}, Proposition~8.10.3 on p.~186.)

  First we prove~\eqref{eq:CAPM1}.
  Without loss of generality, assume $0 < \epsilon < 1$.
  For any $\alpha>0$,
  Speculator has a trivial strategy
  witnessing that the EMH for $m$ predicts
  \begin{equation}\label{eq:auxgoal1}
    \prod_{n=1}^N
    \left(
      1 + \epsilon s_n + (1-\epsilon) m_n
    \right)
    <
    \frac{1}{\alpha}
    \prod_{n=1}^N
    \left(
      1 + m_n
    \right)
  \end{equation}
  at level $\alpha$:
  On each round, he invests $\epsilon$ of his capital in $s$
  (Investor's portfolio)
  and $1-\epsilon$ of his capital in $m$.
  But we can rewrite~\eqref{eq:auxgoal1} as
  \[
    \sum_{n=1}^N
    \biggl(
      \ln
      \left(
        1 + \epsilon s_n + (1-\epsilon) m_n
      \right)
      -
      \ln
      \left(
        1 + m_n
      \right)
    \biggr)
    <
    \ln \frac{1}{\alpha}.
  \]
  This implies
  \begin{multline*}
    \sum_{n=1}^N
    \biggl(
      \epsilon s_n + (1-\epsilon) m_n
      -
      \frac12 \epsilon^2 s_n^2 - \frac12 (1-\epsilon)^2 m_n^2
      - \epsilon (1-\epsilon) s_n m_n \\
      + \gamma(\epsilon s_n + (1-\epsilon) m_n)
      -
      m_n
      + \frac12 m_n^2
      - \Gamma(m_n)
    \biggr)
    <
    \ln\frac{1}{\alpha}
  \end{multline*}
  or
  \begin{multline*}
    \sum_{n=1}^N
    \biggl(
      \epsilon
      \left(
        s_n - m_n + m_n^2 - s_n m_n
      \right)
      -
      \frac12 \epsilon^2
      \left(
        s_n^2 + m_n^2 - 2 s_n m_n
      \right) \\
      + \gamma(s_n \wedge m_n)
      - \Gamma(m_n)
    \biggr)
    <
    \ln\frac{1}{\alpha}
  \end{multline*}
  or
  \begin{multline*}
    \epsilon
    \left(
      \mu_s - \mu_m + \sigma_m^2 - \sigma_{sm}
    \right)
    -
    \frac12 \epsilon^2
    \left(
      \sigma_s^2 + \sigma_m^2 - 2 \sigma_{sm}
    \right) \\
    <
    \frac1T
    \sum_{n=1}^N
    \left|\gamma(s_n)\right|
    +
    \frac1T
    \sum_{n=1}^N
    \left|\gamma(m_n)\right|
    +
    \frac1T
    \sum_{n=1}^N
    \Gamma(m_n)
    +
    \frac1T\ln\frac{1}{\alpha}
  \end{multline*}
  or
  \begin{multline*}
    \mu_s - \mu_m + \sigma_m^2 - \sigma_{sm}
    <
    \frac12 \epsilon
    \sigma_{s-m}^2 \\
    +
    \frac{1}{3\epsilon}
    \sigma_s^2
    \max_n \frac{|s_n|}{|1+s_n|^3}
    +
    \frac{1}{3\epsilon}
    \sigma_m^2
    \max_n \frac{|m_n|}{|1+m_n|^3}
    +
    \frac{1}{3\epsilon}
    \sigma_m^2
    \max_n |m_n|
    +
    \frac{1}{T\epsilon}\ln\frac{1}{\alpha}.
  \end{multline*}
  This completes the proof of~\eqref{eq:CAPM1}
  since $\max_n |s_n|$ and $\max_n |m_n|$ are infinitesimal
  and $T$ is infinite.

  Now we prove~\eqref{eq:CAPM2}.
  Without loss of generality, assume $\epsilon<1/3$.
  Consider a strategy for Speculator
  that calls for investing $-\epsilon$ of his capital in $s$
  and investing $1+\epsilon$ of his capital in $m$
  on every round.
  Using the inequalities $m_n\ge-1/2$ and $s_n\le1$
  (remember that $\max_n |m_n|$ and $\max_n |s_n|$ are infinitesimal),
  we can see that this strategy's return on round $n$ is
  \[
    -\epsilon s_n + (1+\epsilon) m_n
    \ge
    -\epsilon + (1+\epsilon)\left(-\frac12\right)
    >
    -1
  \]
  and so it does not risk bankruptcy for Speculator.
  It witnesses that the EMH for $m$ predicts
  \begin{equation}\label{eq:auxgoal2}
    \prod_{n=1}^N
    \left(
      1 - \epsilon s_n + (1+\epsilon) m_n
    \right)
    <
    \frac{1}{\alpha}
    \prod_{n=1}^N
    \left(
      1 + m_n
    \right)
  \end{equation}
  at level $\alpha$,
  and~\eqref{eq:auxgoal2} can be transformed
  (analogously to~\eqref{eq:auxgoal1} with $\epsilon$ replaced by $-\epsilon$)
  as follows:
  \begin{gather*}
    \sum_{n=1}^N
    \biggl(
      \ln
      \left(
        1 - \epsilon s_n + (1+\epsilon) m_n
      \right)
      -
      \ln
      \left(
        1 + m_n
      \right)
    \biggr)
    <
    \ln\frac{1}{\alpha}; \\[1mm]
    \sum_{n=1}^N
    \biggl(
      -\epsilon s_n + (1+\epsilon) m_n
      - \frac12 \epsilon^2 s_n^2 - \frac12 (1+\epsilon)^2 m_n^2
      + \epsilon (1+\epsilon) s_n m_n \\
      \gamma(-\epsilon s_n + (1+\epsilon) m_n)
      -
      m_n
      + \frac12 m_n^2
      - \Gamma(m_n)
    \biggr)
    <
    \ln\frac{1}{\alpha}; \\[1mm]
    \sum_{n=1}^N
    \biggl(
      -\epsilon
      \left(
        s_n - m_n + m_n^2 - s_n m_n
      \right)
      -
      \frac12 \epsilon^2
      \left(
        s_n^2 + m_n^2 - 2 s_n m_n
      \right) \\
      + \gamma\left(m_n \wedge (2m_n - s_n)\right)
      -
      \Gamma(m_n)
    \biggr)
    <
    \ln\frac{1}{\alpha}; \\[1mm]
    -\epsilon
    \left(
      \mu_s - \mu_m + \sigma_m^2 - \sigma_{sm}
    \right)
    -
    \frac12 \epsilon^2
    \sigma_{s-m}^2 \\
    <
    \frac1T \sum_{n=1}^N \left|\gamma(m_n)\right|
    +
    \frac1T \sum_{n=1}^N \left|\gamma(2m_n-s_n)\right|
    +
    \frac1T \sum_{n=1}^N \Gamma(m_n)
    +
    \frac1T \ln\frac{1}{\alpha}; \\[1mm]
    \mu_s - \mu_m + \sigma_m^2 - \sigma_{sm}
    >
    -\frac12 \epsilon \sigma_{s-m}^2 \\
    -
    \frac{1}{3\epsilon}
    \sigma_m^2
    \max_n \frac{|m_n|}{|1+m_n|^3}
    -
    \frac{1}{3\epsilon}
    \sigma_{2m-s}^2
    \max_n \frac{|2m_n-s_n|}{|1+2m_n-s_n|^3}
    -
    \frac{1}{3\epsilon}
    \sigma_m^2
    \max_n |m_n| \\
    -
    \frac{1}{T\epsilon}
    \ln\frac{1}{\alpha}
  \end{gather*}
  (we have used the obvious notation $\sigma_{2m-s}^2$).
  This proves~\eqref{eq:CAPM2}.
\end{proof}

\begin{proof}[Proof of Proposition~\ref{prop:2}]
  We can bound the left-hand side of the inequality in Proposition~\ref{prop:2}
  from above as follows:
  \begin{multline*}
    \lambda_s - \lambda_m + \frac12 \sigma_{s-m}^2 \\
    \le
    \left(
      \mu_s - \frac12 \sigma_s^2
      +
      \frac1T \sum_{n=1}^N \Gamma(s_n)
    \right)
    -
    \left(
      \mu_m - \frac12 \sigma_m^2
      +
      \frac1T \sum_{n=1}^N \gamma(m_n)
    \right) \\
    +
    \frac{\sigma_s^2}{2}
    +
    \frac{\sigma_m^2}{2}
    -
    \sigma_{sm} \\
    \le
    \mu_s - \mu_m + \sigma_m^2 - \sigma_{sm}
    +
    \frac1T \sum_{n=1}^N \Gamma(s_n)
    +
    \frac1T \sum_{n=1}^N \left|\gamma(m_n)\right|;
  \end{multline*}
  combining this with Proposition~\ref{prop:1},
  we can see that the EMH for $m$ predicts
  \[
    \lambda_s
    -
    \lambda_m
    +
    \frac12 \sigma_{s-m}^2
    <
    \epsilon (1 + \sigma_{s-m}^2),
  \]
  for any $\epsilon>0$.

  In the same way, we can bound the left-hand side
  of the inequality in Proposition~\ref{prop:2} from below:
  \begin{multline*}
    \lambda_s - \lambda_m + \frac12 \sigma_{s-m}^2 \\
    \ge
    \left(
      \mu_s - \frac12 \sigma_s^2
      +
      \frac1T \sum_{n=1}^N \gamma(s_n)
    \right)
    -
    \left(
      \mu_m - \frac12 \sigma_m^2
      +
      \frac1T \sum_{n=1}^N \Gamma(m_n)
    \right) \\
    +
    \frac{\sigma_s^2}{2}
    +
    \frac{\sigma_m^2}{2}
    -
    \sigma_{sm} \\
    \ge
    \mu_s - \mu_m + \sigma_m^2 - \sigma_{sm}
    -
    \frac1T \sum_{n=1}^N \left|\gamma(s_n)\right|
    -
    \frac1T \sum_{n=1}^N \Gamma(m_n);
  \end{multline*}
  combining this with Proposition~\ref{prop:1},
  we can see that the EMH for $m$ predicts
  \[
    \lambda_s
    -
    \lambda_m
    +
    \frac12 \sigma_{s-m}^2
    >
    -\epsilon (1 + \sigma_{s-m}^2),
  \]
  for any $\epsilon>0$.
  This completes the proof.
\end{proof}

\end{document}